\documentclass[10pt, conference, compsocconf]{IEEEtran}
\usepackage{cite}
\usepackage{balance}
\usepackage{color}
\usepackage{epsfig}
\usepackage{epstopdf}
\usepackage{graphicx}

\usepackage{amssymb,amsmath,amsthm}
\newtheorem{lemma}{Lemma}

\usepackage{bm}
\usepackage{algorithm}
\usepackage[noend]{algpseudocode}
\usepackage{pict2e}
\makeatletter
\def\BState{\State\hskip-\ALG@thistlm}
\makeatother
\newcommand{\comment}[1]{}

\usepackage{colortbl}

\begin{document}
\title{Dynamic Radio Cooperation for Downlink Cloud-RANs with Computing Resource Sharing}

\author{\IEEEauthorblockN{Tuyen X. Tran and Dario Pompili}
\IEEEauthorblockA{Department of Electrical and Computer Engineering\\
Rutgers University--New Brunswick, NJ, USA\\
E-mails: \{tuyen.tran, pompili\}@cac.rutgers.edu}
}

\maketitle

\thispagestyle{empty}

\pagestyle{plain} \pagenumbering{arabic}

\begin{abstract} 
A novel dynamic radio-cooperation strategy is proposed for Cloud Radio Access Networks (C-RANs) consisting of multiple Remote Radio Heads (RRHs) connected to a central Virtual Base Station (VBS) pool. In particular, the key capabilities of C-RANs in computing-resource sharing and real-time communication among the VBSs are leveraged to design a joint dynamic radio clustering and cooperative beamforming scheme that maximizes the downlink weighted sum-rate system utility (WSRSU). Due to the combinatorial nature of the radio clustering process and the non-convexity of the cooperative beamforming design, the underlying optimization problem is NP-hard, and is extremely difficult to solve for a large network. Our approach aims for a suboptimal solution by transforming the original problem into a Mixed-Integer Second-Order Cone Program (MI-SOCP), which can be solved efficiently using a proposed iterative algorithm. Numerical simulation results show that our low-complexity algorithm provides close-to-optimal performance in terms of WSRSU while significantly outperforming conventional radio clustering and beamforming schemes. Additionally, the results also demonstrate the significant improvement in computing-resource utilization of C-RANs over traditional RANs with distributed computing resources.

\end{abstract}

\section{Introduction}\label{sec:intro}

{\bf{Overview:}} The proliferation of personal mobile-computing devices along with a plethora of data-intensive mobile applications has resulted in a tremendous increase in demand for ubiquitous and high-data-rate wireless communications over the last few years. To cope with this challenge, the current trend in cellular networks is to increase the densification of small cells and to leverage the cooperation among multiple antennae and base stations (BSs). In this way, higher system throughput and reduced interference can be achieved via Coordinated Multi-Point (CoMP) transmission and reception techniques\cite{lee2012comp}, which have been adopted in 3GPP Long-Term Evolution (LTE)-Advanced. In CoMP, a set of neighboring cells are grouped into clusters, each consisting of connected BSs that share Channel State Information (CSI) and user signals. This scheme allows for joint processing among BSs that can effectively mitigate the Inter-Cell Interference (ICI) and thus improve the spectral efficiency. However, in current cellular-network architectures, physical links only exist between BSs and their corresponding access network gateway and thus, the control signaling between BSs needed to realize CoMP has to travel through costly backhaul links, and often over a one-level higher layer in the aggregation hierarchy. Consequently, the latency and scarce interconnection capacity among BSs have resulted in limited deployments of CoMP in practice and, in turn, in modest BS cooperation.

\begin{figure}
 \centering
\includegraphics[scale = 0.4]{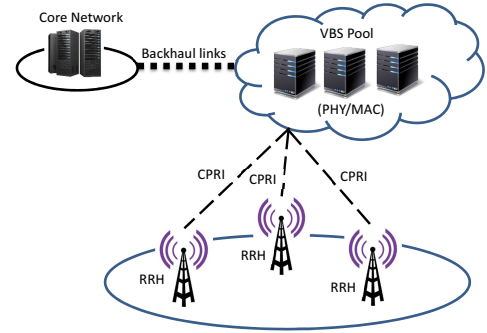} 
\caption{Cloud Radio Access Network (C-RAN) Architecture.}\label{fig:cran_nw}
\end{figure}

%\begin{figure*}
% \centering
%\begin{tabular}{cc}
%\includegraphics[scale = 0.6]{fig/static.eps} \hspace*{2cm}& \includegraphics[scale=0.6]{fig/dynamic.eps} \\
%\small (a) Static Clustering   \hspace*{2cm} &\small(b) Dynamic Clustering
%\end{tabular}
%\caption{Cooperative radio clustering; (a) Static clustering with fixed non-overlapping clusters and (b) Dynamic clustering where each Remote Radio Head (RRH) cluster is formed for a user.}\label{fig:clustering}
%\end{figure*}

%\begin{figure*}[t!]
%\centering
%\includegraphics[scale=0.8]{NAF_Protocol_System_Journal.eps}
%\caption{(a) The half-duplex NAF multi-relay system.  (b) The frame structure of the NAF multi-relay protocol; the solid boxes denote the transmitted signals and the dashed boxes denote the received signals.}\label{fig:sys}
%\end{figure*}

Recently, Cloud Radio Access Network (C-RAN)~\cite{whitepaper13, sundaresan2013cloud, Pompili2015} architecture has been introduced as a new paradigm for broadband wireless access that allows for dynamic reconfiguration of computing resources and provides a higher degree of cooperation as well as communication among the BSs. The fundamental characteristics of C-RAN can be summarized as i) centralized management of computing resources, ii) reconfigurability of spectrum resources, iii) collaborative communications, and iv) real-time cloud computing on generic platforms. A typical C-RAN, as shown in Fig.~\ref{fig:cran_nw}, is composed of three main parts: 1) Remote Radio Heads (RRHs) plus antennae, which are located at the cell sites and are controlled remotely by Virtual Base Stations (VBSs) housed in a centralized VBS pool, 2) the Base Band Unit (BBU) (VBS pool) composed of high-speed programmable processors and real-time virtualization technology to carry out digital processing tasks, 3) low-latency, high-bandwidth Common Public Radio Interface (CPRI), which connects the RRHs to the VBS pool. 

In this paper, we aim to realize the benefits offered by C-RANs to improve the cellular network performance via dynamic adaptation of radio clusters and computing resources. Firstly, the co-location model of the VBSs allows for their real-time intercommunication, thus fully enabling a coordinated joint transmission of RRHs that is currently practically constrained. In particular, control signals to realize CoMP between the BSs that traditionally travel via back-haul links can now be exchanged through the InfiniBand interconnection among the VBSs. A C-RAN-based radio-cooperation scheme would be fully dynamic and user specific, in the sense that we can form a virtual cluster of RRHs to coordinate their downlink transmissions to each of the scheduled users. In this strategy, each scheduled user is always the central of a RRH cluster, making it different from the traditional CoMP techniques where the RRHs are grouped into fixed and non-overlapping clusters.

%For example, Fig.~\ref{fig:clustering}(a) depicts a traditional \emph{static} clustering scheme where six RRHs are grouped into two distinct clusters, i.e., $V_1=\left\{R_1,R_4,R_5\right\}$ and $V_2=\left\{R_2,R_3,R_6\right\}$. Users falling within each cluster, in this case user $U_1$ in $V_1$ and user $U_2$ in $V_2$, can only receive signals from the RRHs belonging to that cluster. On the other hand, in the \emph{dynamic} clustering scheme depicted in Fig.~\ref{fig:clustering}(b), the user-specific RRH clusters are optimally formed to serve each user and there is no explicit cluster boundaries. 

{\bf{Related Works:}}
Pioneering works on realizing the benefit of C-RANs have focused on the overall system architecture with emphasis on system issues, feasibility of virtual software base station stacks, performance requirements and analysis of optical links between RRHs and their VBSs \cite{whitepaper13, bhaumik2012cloudiq, madhavan2012quantifying}. 
%In , the authors introduce the new centralized base station idea and its advantages, challenges and requirements. In \cite{gupta2010unlocking}, Gupta et al. address the problem of inter-BS communication and the latencies involved in information exchange among distributed base stations and consider base station pooling as a potential solution for higher degree of base station co-operation in broadband cellular networks. In \cite{zhu2011virtual}, the authors describe the virtual base station pool structure and important system challenges in implementing this concept on IT platforms.
On the other hand, considerable attention has also been paid on cooperative communications techniques for C-RAN under various different objectives. For example, in \cite{hajisamicocktail} the authors propose a blind source separation strategy to mitigate interference in uplink C-RAN; in \cite{shi2014group, Luo2015downlink} the authors consider a network power minimization problem. In addition, the optimal tradeoff design between transmit power and backhaul capacity is studied in \cite{ha2014energy}, while the tradeoff between transmit power and delay performance is investigated in \cite{li2014resource} via a cross-layer based approach.  

In this paper, we study a dynamic radio cooperation technique and consider Weighted Sum-Rate System Utility (WSRSU) as the performance metric under a practical constraint on computing resources at the VBS pool. Note that the BS cooperation for WSRSU maximization problem has been studied in traditional CoMP systems. However, due the scarce interconnection among the BSs and the lack of global CSI available at each BS, existing clustering and cooperative beamforming techniques are mostly heuristic-based (i.e., the clustering decision is made based on the relative signal strength and locations of the users, and the beamforming design is not adaptive to \emph{inter-cluster interference}) \cite{gong2011joint, huang2009increasing, baracca2012dynamic}.

{\bf{Our Contributions:}}
In this paper, we propose a novel dynamic radio cooperation strategy for C-RANs that takes advantage of real-time communication and computing-resource sharing among the VBSs. Unlike existing methods, our approach makes the joint clustering and beamforming decision based on global CSI available at the VBS pool, thus being able to mitigate both the \emph{intra-cluster} and \emph{inter-cluster} interference in order to significantly improve the system's performance. Our proposed solution dynamically groups the RRHs into user-specific (potentially overlapping) clusters and designs the downlink beamformers at each RRH in order to maximize the WSRSU function. In particular, within each scheduling interval, i.e., a time-frequency resource block, a group of RRHs is identified to serve each scheduled user. To realize the proposed \emph{Dynamic Radio Cooperation (Dynamic-RC)} strategy, we formulate the associated optimization problem, which we also refer to as the \emph{Dynamic-RC problem}, that aims to maximize the WSRSU under the transmit power constraints at the RRHs and the total computing-resource constraint at the VBS pool. Due to the combinatorial nature of the radio clustering process and the non-convexity of the cooperative beamforming design, the \textit{Dynamic-RC} problem is extremely difficult to solve optimally in practical (polynomial) time for a system with a large number of users and RRHs. To overcome this drawback and solve the problem efficiently, our approach aims for a suboptimal solution with reasonable complexity. In particular, we exploit \textit{conic programming} techniques~\cite{alizadeh2003second} and the $l1$-norm reweighting approximation methods from Compressive Sensing  which were originally proposed for sparse signal recovery \cite{l1_reweighted}, in order to quickly identify the optimal clustering decision and beamforming design.

We propose an iterative algorithm to solve the \emph{Dynamic-RC} problem. In each iteration, the clustering decision is temporarily fixed and a \emph{Cooperative Beamforming Design} (CBD) problem is solved using Second-Order Cone Programming (SOCP) technique. The optimal beamforming solution obtained from the CBD problem is used to adjust the clustering decision via $l1$-norm reweighting technique. As such, the joint clustering and beamforming design is quickly identified and is adaptive to the global network condition. 

Numerical simulations are carried out extensively in various user distribution scenarios and demonstrate that our proposed low-complexity \emph{Dynamic-RC} strategy significantly improves the WSRSU performance over conventional radio clustering and beamforming schemes. Furthermore, the results also show the great potential gains of C-RANs using our \emph{Dynamic-RC} strategy over distributed RANs in terms of computing resource and transmit power utilization. 

{\bf{Paper Organization:}}
The remainder of this paper is organized as follows: in Sect.~\ref{sec:formulation}, we present the considered system model and formulate the problem under study; in Sect.~\ref{sec:design1}, we discuss the analysis and solution to the cooperative beamforming design problem with a fixed clustering decision; in Sect.~\ref{sec:design2}, the \textit{Dynamic-RC} strategy via dynamic radio clustering and beamforming design is solved via our proposed iterative algorithm; simulation results are illustrated in Sect.~\ref{sec:result} and, finally, Sect.~\ref{sec:con} concludes the paper and points to future work.

\section{System Model and Problem Formulation}\label{sec:formulation}
%\begin{figure}[htb!]
%\centering
%\includegraphics[scale=0.6]{Figures/new.eps}
%\caption{C-RAN Architecture}
%\label{fig:sys}
%\end{figure}
In this section, we firstly introduce the system model of the considered downlink C-RAN system and discuss the computing-resource constraint. The proposed dynamic radio cooperation strategy is then formulated as a joint clustering and beamforming design problem. 

\subsection{System Model}
We consider a multi-user, multi-cell C-RAN downlink system, where each cell has one RRH that connects to a common VBS pool via high-capacity backhaul links. Let $\mathcal{R} = \{1,2,...,R\}$ be the set of RRHs and $\mathcal{U}=\{1,2,...,U\}$ be the set of active users in the system. We assume that each RRH $r$ has $N_r$ antennae while, realistically, all the users are equipped with only a single antenna. Note that the solutions proposed can be trivially extended to the multi-antenna-user case. The RRHs cooperate with each other to form virtual user-specific clusters, i.e., each RRH cluster is formed for a scheduled user, while each RRH can be part of multiple clusters. Hence, the number of virtual clusters is equal to the number of scheduled users in the system, which may be smaller than the number of total active users. Let $\mathcal{S}=\left\{ {s_u^r\left| {u \in \mathcal{U}, r \in \mathcal{R}} \right.} \right\}$ be the clustering decision, in which $s_u^r$ is a binary variable equal to $1$ if RRH $r$ is selected to serve user $u$, and $0$ otherwise. Consequently, let ${\mathcal{V}_u} = \left\{ {r \in \mathcal{R}\left| {s_u^r = 1} \right.} \right\}$ denote the serving cluster of user $u$. We consider the system in a single time-frequency resource block, which is considered to be spatially reused across all the users. As such, each RRH can simultaneously serve at most $N_r$ users; otherwise, the users will suffer from intra-cluster interference. 

We assume that each user has a single traffic flow that is independent of all other users' flows. Baseband signals for user $u$  and the corresponding downlink beamforming information after being processed at the VBS pool will be transported to all the RRHs in the serving cluster $\mathcal{V}_u$. In each scheduling slot, all the RRHs in $\mathcal{V}_u$ will jointly transmit the normalized symbol $x_u\in \mathbb{C}$ of unit power to user $u$. It is assumed that the signals for different users are independent from each other and from the receiver noise. Now, let ${\bf{w}}_u^r \in {\mathbb{C}^{{N_r} \times 1}}$ be the linear downlink beamforming vector at RRH $r$ corresponding to user $u$ and ${\bf{W}} = \left\{ {{\bf{w}}_u^r \left| \forall u \in \mathcal{U},r \in \mathcal{R}\right.} \right\}$ denote the network beamforming design. Note that ${\bf{W}}$ also implies the scheduling decision, i.e., user $u$ is not scheduled for the current time-frequency slot if ${\bf{w}}_u^r = {\bf{0}}, \forall r \in \mathcal{R}$. In the current scheduling slot, the received signal $y_u \in \mathbb{C}$ at user $u$ is, 
\begin{equation} \label{eq:yu}
{y_u} = \underbrace {\sum\limits_{r \in {{\cal V}_u}} {{\bf{h}}_u^r{\bf{w}}_u^r{x_u}} }_{{\rm{desired}}{\kern 1pt} {\rm{signal}}} + \underbrace {\sum\limits_{u' \in {\cal U},u' \ne u} {\sum\limits_{r' \in {{\cal V}_{u'}}} {{\bf{h}}_u^{r'}{\bf{w}}_{u'}^{r'}{x_{u'}}} } }_{{\rm{interference}}} + {z_u},
\end{equation}
where ${\bf{h}}_u^r \in {\mathbb{C}^{1 \times {N_r}}}$ is the channel coefficient vector from RRH $r$ to user $u$, $z_u$ is the zero-mean circularly symmetric Gaussian noise denoted as $\mathcal{CN}(0, \sigma^2)$. For simplicity, let ${\Psi _{u,u'}} = \sum\limits_{r' \in {\mathcal{V}_{u'}}} {{\bf{h}}_u^{r'}{\bf{w}}_{u'}^{r'}}$ and ${\Psi _u} = {\Psi _{u,u}} = \sum\limits_{r \in {\mathcal{V}_{u}}} {{\bf{h}}_u^{r}{\bf{w}}_{u}^{r}}$. With this position, the received Signal-to-Interference-plus-Noise Ratio (SINR) at user $u$ is,
\begin{equation} \label{eq:gamma_u}
{\gamma _u} = \frac{{{{\left| {{\Psi _u}} \right|}^2}}}{{\sum\limits_{u' \in {\cal U} ,u' \ne u} {{{\left| {{\Psi _{u,u'}}} \right|}^2} + {\sigma ^2}} }}.
\end{equation}

Thus, under the clustering decision ${\cal S}$ and the beamforming design ${\bf{W}}$, the Shannon transmission rate of user $u$ can be calculated as ${R_u}\left( {{\cal S},{\bf{W}}} \right) = \eta B{\log _2}\left( {1 + \mu {\gamma _u}} \right)$,
 in which $B~\mathrm{[Hz]}$ is the channel bandwidth and $\eta ,\mu  \in [0,1]$ account for the spectral and the coding efficiencies, respectively.  Unless otherwise stated, for notation simplicity in the subsequent analysis we will assume $B=\eta=\mu=1$ and consider the normalized rate ($\mathrm{bits/s/Hz}$). Hence, the rate $R_u$ simplifies to,
 \vspace*{-0.1cm}
\begin{equation} \label{eq:Ru}
{R_u}\left( {{\cal S},{\bf{W}}} \right) = {\log _2}\left( {1 + {\gamma _u}} \right).
\end{equation}

\textit{\underline{Computing resource constraint:}} The VBS pool consists of a set of interconnected VBSs hosted in the physical-server infrastructure of a datacenter. Each VBS performs baseband processing for a certain set of users, and by leveraging virtualization technology, these VBSs can flexibly share the common computing resource of the physical server pool. Recently, the implementation of software VBSs on General-Purpose Platform (GPP) has been realized (see, for example,~\cite{bhaumik2012cloudiq, madhavan2012quantifying}). Profiling results on these systems have revealed that the utilized computing resource at a VBS is an increasing function of the accumulated data rates processed by that VBS. Therefore, it is reasonable to argue that the total computing-resource capacity of the VBS pool places a \emph{cap} on the total data rates of the users in the network. In general, the computing-resource capacity of the VBS pool can be modeled as a multi-dimentional vector representing the capacities of the CPUs, memory, and network interfaces. However, for the ease of analysis, we only consider scalar computing capacity in this paper. In particular, let $C$ denote the total computing capacity in the VBS pool that can be flexibly shared among all the VBSs. The computing-resource constraint on the accumulated data rate of all the users in the system can be expressed as 
\vspace*{-0.1cm}
\begin{equation} \label{eq:centralized}
\Gamma \left(  {\sum\limits_{u \in \mathcal{U}} {{R_u}} } \right) \le C,
\end{equation}
%\vspace*{-0.5cm}
where $R_u$ is the data rate of user $u$ given in \eqref{eq:Ru} and $\Gamma(.)$ is an increasing function specifying the relationship between the utilized computing resource and the accumulated user data rate\footnote{The realization of $\Gamma(.)$ can be obtained by carefully profiling the VBSs at different level of offered load in a practical C-RAN implementation.}. It should be noted that for a traditional system with distributed computing resource at the RRHs, the accumulated data rate processed at each RRH $r$ will be subject to the per-RRH computing-resource constraint $C_r<C$, i.e., 
\vspace*{-0.1cm}
\begin{equation} \label{eq:distributed}
\Gamma \left( {\sum\limits_{u \in {\mathcal{U}}} {{s_u^rR_u}} } \right) \le {C_r}, \forall r \in \mathcal{R}.
\end{equation}

\subsection{Joint Clustering and Beamforming Problem Formulation}
Our objective is to maximize the WSRSU under the transmit power constraint at each RRH and the total computing-resource constraint at the VBS pool. It is assumed that the capacity of the front-haul links connecting RRHs to the VBS pool is sufficiently provisioned to accommodate peak-capacity demand. Our proposed dynamic radio cooperation strategy involves finding the optimal clustering decision $\mathcal{S}^*$ and the optimal beamforming design $\mathbf{W}^*$, and can be formulated as,
\begin{subequations} \label{eq:prob1}
\begin{align}
\label{eq:prob1a}
\left( {\mathcal{S}^*,{\bf{W}}^*} \right) = \mathop  {\text{argmax}}\limits_{\scriptstyle\left\{ {s_u^r,{\bf{w}}_u^r} \right\}\hfill\atop
\scriptstyle r \in \mathcal{R},u \in \mathcal{U}\hfill} &\sum\limits_{u \in \mathcal{U}} {{q_u}{R_u}\left( {{\cal S},{\bf{W}}} \right)} \\
\label{eq:prob1b}
\text{s.t.} \hspace{0.5cm} &\sum\limits_{u \in \mathcal{U}} {\left\| {{\bf{w}}_u^r} \right\|_2^2}  \le {P_r},\forall r \in \mathcal{R},\\
\label{eq:prob1b_1}
&\left\| {{\bf{w}}_u^r} \right\|_2^2 \le s_u^r{P_r},\\
\label{eq:prob1c}
&\sum\limits_{u \in \mathcal{U}} {{R_u}}\left( {{\cal S},{\bf{W}}} \right)  \le \Omega ,\\
\label{eq:prob1d}
&\sum\limits_{u \in \mathcal{U}} {s_u^r}  \le {N_r}, s_u^r \in \left\{ {0,1} \right\},
\end{align}
\end{subequations}where $q_u$, $u \in \mathcal{U}$, is the utility marginal function corresponding to user $u$, which can represent the user-specific Quality of Service (QoS) or priority in the system, $P_r~\mathrm{[W]}$ is the per-RRH transmission power constraint and $\Omega =\arg \Gamma(C)$. Constraint~\eqref{eq:prob1b_1} indicates the coupling between the assignment variable $s_u^r$ and the beamforming vector ${\bf{w}}_u^r$, i.e., ${\bf{w}}_u^r = {\bf{0}}$ when $s_u^r=0$.
% Note that the beamformer variable ${\bf{w}}_u^r$  and the assignment variable $s_u^r$ are coupled as
%\begin{equation}
%s_u^r = \left\{ {\begin{array}{*{20}{c}}
%0&{{\bf{w}}_u^r = {\bf{0}},}\\
%1&{\text{otherwise},}
%\end{array}} \right.
%\end{equation}
We refer to~\eqref{eq:prob1} as the dynamic radio cooperation (\textit{Dynamic-RC}) problem. In fact, this is a Mixed-Integer Non-Linear Program (MINLP), which is intractable in practical time. Specifically, even when the binary variables $s_u^r$ are fixed, solving for ${\bf{w}}_u^r$ is still NP-hard. 

Given a large number of variables that scales linearly with the number of users and RRHs in the system, finding a low-complexity, suboptimal solution is highly desirable. To this end, we firstly solve the Cooperative Beamforming Design (CBD) problem with given clustering decision $\mathcal{S}$, and propose a low-complexity iterative algorithm to solve the \textit{Dynamic-RC} problem to a local optimum. Specifically, in Sect.~\ref{sec:design1}, we will transform the CBD problem into a SOCP with a fixed clustering decision, and will take advantage of the existing efficient SOCP algorithms. The \textit{Dynamic-RC} problem will then be solved in Sect.~\ref{sec:design2} using the iterative $l1$-norm reweighting technique, which solves the CBD problem and updates the clustering decision in each iteration.

%+++++++++++++++++++++++++++++++++++++++++++++++++++++++++++++++++++
\vspace*{-.3cm}
\section{Cooperative Beamforming with Fixed Clustering Decision}

\label{sec:design1}
In this section, we consider the problem of Cooperative Beamforming Design~(CBD) for a given radio clustering decision $\mathcal{S}$. In particular, for given $\left\{ {s_u^r} \right\}$ satisfying constraints~\eqref{eq:prob1d}, we need to find the optimal downlink beamformers $\left\{ {{\bf{w}}_u^r} \right\}$ by solving the CBD problem below,
\begin{subequations} \label{eq:static}
\begin{align}
\label{eq:static_a}
\mathop {\max }\limits_{{\bf{w}}_u^r,r \in \mathcal{R},u \in \mathcal{U}}  &\sum\limits_{u \in \mathcal{U}} {{q_u}{R_u}}\left( {{\cal S},{\bf{W}}} \right) \\
\label{eq:static_b}
\text{s.t.} \hspace{0.5cm} &\sum\limits_{u \in \mathcal{U}} {\left\| {{\bf{w}}_u^r} \right\|_2^2}  \le {P_r},\forall r \in \mathcal{R},\\
\label{eq:static_c}
&\sum\limits_{u \in \mathcal{U}} {{R_u}}\left( {{\cal S},{\bf{W}}} \right)  \le \Omega.
\end{align}
\end{subequations}Observe that the rate functions $R_u$'s appear in \emph{both} the constraint and objective of \eqref{eq:static}, making the problem difficult to deal with. To decouple this problem with respect to (w.r.t.) $R_u$'s, we remove the constraint~\eqref{eq:static_c} and consider the \emph{relaxed}-CBD problem with constraint~\eqref{eq:static_b} only. The solution $\left\{ {{\bf{\tilde w}}_u^r} \right\}$ of the \emph{relaxed}-CBD problem will be verified against constraint~\eqref{eq:static_c} so to finally obtain the solution of the original CBD problem by solving an additional \emph{feasibility} problem. In the following subsections, the \emph{relaxed}-CBD first and then the \emph{feasibility} problem will be addressed sequentially.
%To decouple the problem with respect to (w.r.t.) $R_u$'s, we first consider the problem without constraint~\eqref{eq:static_c}, which we refer to as the \textit{power-constrained} CBD~(PC-CBD) problem as the sole constraint left is~\eqref{eq:static_b}. The solution $\left\{ {{\bf{\tilde w}}_u^r} \right\}$ of the PC-CBD problem will be verified against constraint~\eqref{eq:static_c} so to finally obtain the solution of the original CBD problem by solving an additional \emph{feasibility} problem. In the following subsections, the PC-CBD first and then the feasibility problem will be addressed sequentially.

\subsection{Relaxed-CBD Problem}
The \emph{relaxed}-CBD problem is rewritten from \eqref{eq:static} without the computing-capacity constraint~\eqref{eq:static_c}, and is cast as follows,
\begin{subequations} \label{eq:power}
\begin{align}
\label{eq:power_a}
\mathop {\max }\limits_{{\bf{w}}_u^r,r \in \mathcal{R},u \in \mathcal{U}}  &\sum\limits_{u \in \mathcal{U}} {{q_u}{R_u}}\left( {{\cal S},{\bf{W}}} \right) \\
\label{eq:power_b}
\text{s.t.} \hspace{0.5cm} &\sum\limits_{u \in \mathcal{U}} {\left\| {{\bf{w}}_u^r} \right\|_2^2}  \le {P_r},\forall r \in \mathcal{R}.
\end{align}
\end{subequations}This is in fact a weighed sum-rate maximization problem, which is widely known to be NP-hard. Our approach aims for a local solution using a low-complexity algorithm designed by effectively exploiting the techniques of SOCP.\footnote{Second-Order Cone Problems (SOCP) are convex-optimization problems in which a linear function is minimized over the intersection of an affine set and the product of second-order (quadratic) cones.} In order to use the efficient algorithms developed for SOCP, one must reformulate the problem into the standard form that the algorithms (e.g., those proposed in~\cite{tran2012fast}) are capable of dealing with. Firstly, from~\eqref{eq:Ru}, objective function~\eqref{eq:power_a} is rewritten as,
\begin{equation} 
\sum\limits_{u \in \mathcal{U}} {{q_u}{R_u}}\left( {{\cal S},{\bf{W}}} \right)  = \sum\limits_{u \in \mathcal{U}} {{{\log }_2}{{\left( {1 + {\gamma _u}} \right)}^{{q_u}}}}.
\end{equation}Now, by introducing the variables $t_u$'s, $u \in \mathcal{U}$, we can recast the \emph{relaxed}-CBD problem in~\eqref{eq:power} as, 
\begin{subequations}\label{eq:power1}
\begin{align}\label{eq:power1_a}
\mathop {\max }\limits_{{\bf{w}}_u^r,r \in \mathcal{R},u \in \mathcal{U}}  &\prod\limits_{u \in \mathcal{U}} {{t_u}}\\
\label{eq:power1_b}
\text{s.t.} \hspace{0.5cm} &{\gamma _u} \ge t_u^{1/{q_u}} - 1,\forall u \in \mathcal{U}, \\
\label{eq:power1_c}
&\sum\limits_{u \in \mathcal{U}} {\left\| {{\bf{w}}_u^r} \right\|_2^2}  \le {P_r},\forall r \in \mathcal{R},
\end{align}
\end{subequations}which stems from the fact that constraints~\eqref{eq:power1_b} are active at the optimum. We now have the following Lemma. 
\begin{lemma}
Let ${\bf{\tilde w}}_u^r = {\bf{w}}_u^r{e^{j\phi _u^r}}$, where $\phi_u^r$ is the phase rotation such that the imaginary part of ${\bf{h}}_u^r{\bf{\tilde w}}_u^r$ equals to zero, $\forall u \in \mathcal{U}, r \in \mathcal{R}$. If ${{\bf{w}}_u^r}$ is optimal to (\ref{eq:power1}), then ${\bf{\tilde w}}_u^r$ is also optimal.
\end{lemma}

\begin{proof}
We can represent ${\bf{h}}_u^r{\bf{w}}_u^r$ as ${\bf{h}}_u^r{\bf{w}}_u^r = \left| {{\bf{h}}_u^r{\bf{w}}_u^r} \right|{e^{j\theta _u^r}}$. By choosing $\phi_u^r = - \theta_u^r$, we have ${\bf{h}}_u^r{\bf{\tilde w}}_u^r = {\bf{h}}_u^r{\bf{w}}_u^r{e^{j\phi _u^r}} = \left| {{\bf{h}}_u^r{\bf{w}}_u^r} \right|$. Recall $\gamma_u$ given in~\eqref{eq:gamma_u}, it is straightforward to verify that substituting ${\bf{w}}_u^r$ by ${\bf{\tilde w}}_u^r$, $\forall u \in \mathcal{U}, r \in \mathcal{R}$, into (\ref{eq:power1}) will result in the same objective function and constraints. Thus, if ${\bf{w}}_u^r$ is optimal then ${\bf{\tilde w}}_u^r$ is also optimal.
\end{proof}

%\begin{lemma}
%Let ${\bf{\tilde w}}_u^r = diag\left\{ {{e^{j{\phi _i}}}} \right\}{\bf{w}}_u^r$, where $\phi_i, i=1,...,N_r$, are the phase rotations such that the imaginary part of ${{{\left( {{\bf{h}}_u^r} \right)}^H}{\bf{\tilde w}}_u^r}$ equals to zero, $\forall u \in \mathcal{U}, r \in \mathcal{R}$. If $\left\{ {{\bf{w}}_u^r} \right\}$ is the optimal solution of (\ref{eq:power1}), then ${\bf{\tilde w}}_u^r$ is also an optimal solution of (\ref{eq:power1}).
%\end{lemma}
%\begin{proof}
%s
%\end{proof}
%With $\gamma_u$ given in~\eqref{eq:gamma_u}, and using the definition of $\Psi_u$ and $\Phi \left( {u,u'} \right)$ prior to that, it can be seen that forcing the imaginary part of ${{{\left( {\bf{h}}_u^r \right)}^H}{\bf{w}}_u^r}$ to zero does not affect the optimality of~\eqref{eq:power1}, i.e., if ${{\bf{w}}_u^r}$ is optimal, then ${\bf{w}}_u^r \textit{diag} \left\{ {{e^{j\phi_i}}} \right\}$, where $\phi_i, i=1,...,N_r$, are arbitrary phases, is \emph{also} optimal. This is straightforward to verify since a phase rotation on ${\bf{w}}_u^r$ will result in the same objective function and constraints of~\eqref{eq:power1}

Using Lemma 1, we can restrict ourselves to the beamformers in which ${\bf{h}}_u^r{\bf{w}}_u^r \ge 0$, $\forall u \in \mathcal{U}, r \in \mathcal{V}_u$, where each product has a non-negative real part and a zero imaginary part. Notice that constraint~\eqref{eq:power1_b} is equivalent to
\begin{equation}
\frac{{{{\left| {{\Psi _u}} \right|}^2}}}{{\sum\limits_{u' \in {\cal U} ,u' \ne u} {{{\left| {{\Psi _{u,u'}}} \right|}^2} + {\sigma ^2}} }} \ge t_u^{1/{q_u}} - 1, \forall u \in \mathcal{U},
\end{equation}which can be recast as,
%\begin{align} \label{eq:constraint1}
%&\Psi_u \ge {\beta _u}\sqrt {t_u^{1/{q_u}} - 1} ,\forall u \in \mathcal{U},\\ \label{eq:constraint2}
%\hspace*{-0.5cm} \text{and} \hspace*{0.5cm} &{\Bigg( {\sum\limits_{u' \in \mathcal{U},u' \ne u} {{{\left| {\Phi \left( {u,u'} \right)} \right|}^2} + {\sigma ^2}} } \Bigg)^{1/2} } \le {\beta _u},\forall u \in \mathcal{U},
%\end{align}
\begin{align} \label{eq:constraint1}
&\Psi_u \ge {\beta _u}\sqrt {t_u^{1/{q_u}} - 1} ,\forall u \in \mathcal{U},\\ \label{eq:constraint2}
 \hspace*{-0.5cm} \text{and} \hspace*{0.5cm} &\sqrt{{{\sum\limits_{u' \in \mathcal{U},u' \ne u} {{{\left| {\Psi_{u,u'}} \right|}^2} + {\sigma ^2}} } }} \le {\beta _u},\forall u \in \mathcal{U},
\end{align}by introducing the slack variables $\beta_u$'s and due to the fact that both constraints~\eqref{eq:constraint1} and \eqref{eq:constraint2} are active at the optimum of problem~\eqref{eq:power1}. It can be verified that~\eqref{eq:power1_c} and \eqref{eq:constraint2} follow the Linear Programming (LP) constraint expression with generalized equalities/inequalities, which can be directly written as Second-Order Constraints~(SOCs)\footnote{In a SOC representation, the hyperbolic constraint $ab \geq c^2$, with $a,b \geq 0$, is equivalent to $||{[(a - b)\,\,\,\,\,2c]^T}|{|_2} \le a + b$.}~\cite{alizadeh2003second}. To deal with the non-convex constraint~\eqref{eq:constraint1}, we further exploit the sequential parametric convex-approximation approach in~\cite{beck2010sequential} to approximate~\eqref{eq:constraint1} as convex as presented in the following. 
%\begin{subequations} \label{eq:power2}
%\begin{align}
%\label{eq:power2_b}
%&\sum\limits_{r \in {\mathcal{V}_u}} {{{\left( {{\bf{h}}_u^r} \right)}^H}{\bf{w}}_u^r}  \ge {\beta _u}\sqrt {t_u^{1/{q_u}} - 1} ,\forall u \in \mathcal{U},\\
%\label{eq:power2_c}
%&{\left( {\sum\limits_{\scriptstyle u' \in \mathcal{U}\hfill\atop
%\scriptstyle u' \ne u\hfill} {{{\left| {\sum\limits_{r' \in {\mathcal{V}_{u'}}} {{{\left( {{\bf{h}}_u^{r'}} \right)}^H}{\bf{w}}_{u'}^{r'}} } \right|}^2} + {\sigma ^2}} } \right)^{1/2}} \le {\beta _u},\forall u \in \mathcal{U}.
%\end{align}
%\end{subequations}

Firstly, \eqref{eq:constraint1} can be rewritten as,
\begin{align} \label{eq:16}
&\Psi_u \ge {\beta _u}\sqrt {\xi _u}, \forall u \in \mathcal{U}, \\ \label{eq:17}
&{\xi _u} + 1 \ge t_u^{1/{q_u}}, \forall u \in \mathcal{U}.
\end{align}Observe that, for a given $\phi _u$, we have
\begin{equation} \label{eq:xi}
\beta _u\sqrt {\xi _u}  \le \frac{{{\phi _u}}}{2}\beta _u^2 + \frac{{{\xi _u}}}{{2{\phi _u}}},
\end{equation}which follows the inequality of arithmetic and geometric means of 
${\phi _u}\beta _u^2$ and ${\xi _u}\phi _u^{ - 1}$. The equality in~\eqref{eq:xi} is achieved when $\phi _u = \sqrt {\xi _u} /{\beta _u}$, and we get the equivalent form of constraint~\eqref{eq:16} as,
\begin{equation}  \label{eq:19}
\Psi _u - \frac{{{\xi _u}}}{{2{\phi _u}}} \ge \frac{{{\phi _u}}}{2}\beta _u^2, \forall u \in \mathcal{U}.
\end{equation}

Furthermore, without loss of generality, we scale $q_u$'s in~\eqref{eq:static_a} such that $q_u>1, \forall u \in \mathcal{U}$ to make $t_u^{1/q_u}$ become concave. Thanks to the concavity of $t_u$'s, we can adopt the results in~\cite{beck2010sequential} to replace the right side of~\eqref{eq:17} by its iterative first-order approximation as, 
%
%Following the results in~\cite{beck2010sequential}, we can replace the right side of~\eqref{eq:17} by its first-order approximation due to the concavity of $t_u$'s. In particular, we can iteratively approximate the right side of~\eqref{eq:17} as,
\begin{equation} \label{eq:20}
t_u^{1/{q_u}} \le t_u^{{{(*)}^{1/{q_u}}}} + \frac{1}{{{q_u}}}t_u^{{{(*)}^{\left( {1/{q_u}} \right) - 1}}}\left( {{t_u} - t_u^{(*)}} \right),
\end{equation}where ${t_u^{(*)}}$ denotes the value of $t_u$ in the previous iteration. From~\eqref{eq:constraint2}, \eqref{eq:19}, and \eqref{eq:20}, the \emph{relaxed}-CBD optimization problem in~\eqref{eq:power} can be finally recast as,
\begin{subequations} \label{eq:power2}
\begin{align}
\label{eq:power2_a}
&\mathop {\max }\limits_{{\bf{w}}_u^r,r \in \mathcal{R},u \in \mathcal{U}}  \prod\limits_{u \in \mathcal{U}} {{t_u}}\\
\label{eq:power2_b}
\text{s.t.} \hspace{0.5cm}  &\sum\limits_{u \in \mathcal{U}} {\left\| {{\bf{w}}_u^r} \right\|_2^2}  \le {P_r},\forall r \in \mathcal{R},\\
\label{eq:power2_c}
&(\ref{eq:constraint2}), (\ref{eq:19}), (\ref{eq:20}).
\end{align}
\end{subequations}

Notice that the objective function and all the constraints in~\eqref{eq:power2} admit SOC representation (see~\cite{tran2012fast,alizadeh2003second}). Consequently, the resulting problem in~\eqref{eq:power2} is a SOCP, which can be solved efficiently and very fast using standard solvers such as CPLEX~\cite{cplex} or MOSEK~\cite{mosek}.

\subsection{CBD Feasibility Problem}
Here, the solution of the \emph{relaxed}-CBD problem~\eqref{eq:power} which was obtained via solving the equivalent SOCP problem in~\eqref{eq:power2}, will be verified against the computing-capacity constraint in~\eqref{eq:static_c} to obtain finally the beamforming solution of the original CBD problem cast in~\eqref{eq:static}.

%Suppose that $\left\{ {{\bf{\tilde w}}_u^r} \right\}$ is the solution of problem (\ref{eq:power}) obtained by solving the equivalent SOCP representation as presented in Sect. III-A. Additionally, suppose that 
%$\left\{ {{\bf{w}}_u^{r*}} \right\}$ is the optimal solution that we need to find for problem (\ref{eq:static}), and let $\left\{ {R_u^*} \right\}$ and $\left\{ {{{\tilde R}_u}} \right\}$, $u \in \mathcal{U}$ are the corresponding rate of the beamformers $\left\{ {{\bf{\tilde w}}_u^r} \right\}$ and $\left\{ {{\bf{w}}_u^{r*}} \right\}$, respectively.
%
%Suppose that $\left\{ {{\bf{w}}_u^{r*}} \right\}$ and $\left\{ {R_u^*} \right\}$, $u \in \mathcal{U}$, are the optimal solutions and the corresponding optimal rates of problem (\ref{eq:static}). Similarly, suppose that $\left\{ {{\bf{\tilde w}}_u^r} \right\}$ and $\left\{ {{{\tilde R}_u}} \right\}$, $u \in \mathcal{U}$, are the optimal solutions and the corresponding optimal rates of problem (\ref{eq:power}) obtained by solving the equivalent SOCP representation as presented in Sect. III-A. 

Suppose that ${\bf{\tilde W}}$ is the beamforming solution of problems~\eqref{eq:power}. If ${\bf{\tilde W}}$ satisfies the computing-resource constraint \eqref{eq:static_c}, i.e., $\sum\limits_{u \in \mathcal{U}} {{R_u}\left( {S,{\bf{\tilde W}}} \right)}  \le \Omega $, then ${\bf{\tilde W}}$ is also the optimal solution of ~\eqref{eq:static}. In this case, the WSRSU is limited by the per-RRH power budget only, and not by the computing-resource capacity of the VBS pool. On the other hand, when the computing-resource constraint is violated, we need to selectively drop the rates of some users. This can be done via a greedy algorithm that keeps dropping the users that have the smallest marginal utility function $q_u$ from the current scheduling interval until the total data rate of all the scheduled users satisfies the computing-resource constraint. Since the optimal bearmformer design $\mathbf{W}$ is \emph{jointly} calculated for all users, dropping the rates of some users requires recalculating the beamformers of \emph{all} the RRHs.

%Suppose that ${\bf{W}}^*$ and ${\bf{\tilde W}}$ are the solution of problems~\eqref{eq:static} and \eqref{eq:power}, respectively. Let  $R_u^* = {R_u}\left| {_{{\bf{W}} = {{\bf{W}}^*}}} \right.$ and ${{\tilde R}_u} = {R_u}\left| {_{{\bf{W}} = {\bf{\tilde W}}}} \right.$, $\forall u \in \mathcal{U}$. It is easy to see that if ${\tilde{R}_u}$'s satisfy the computing-capacity constraint in~\eqref{eq:static_c}, i.e., $\sum\limits_{u \in \mathcal{U}} {{{\tilde R}_u}}  \le \Omega $, then ${\tilde{R}_u}$'s are also the solution of~\eqref{eq:static}. In this case, the WSRSU is limited by the per-RRH power budget only, and not by the computing capacity of the VBS pool. On the other hand, when the computing-capacity constraint is violated, we need to selectively drop the rates of some users. This can be done via a greedy algorithm that keeps dropping the users that have the smallest marginal utility function $q_u$ from the current scheduling interval until the total rates of all the scheduled users satisfy the computing-capacity constraint. Since the optimal bearmformer design $\mathbf{W}$ is \emph{jointly} calculated for all users, dropping the rates of some users requires recalculating the beamformers of \emph{all} the RRHs. 

Let $\left\{ {R_u^* \geq 0,u \in \mathcal{U}} \right\}$ be the user rates obtained after the greedy-user-rate-dropping process is applied; the beamformer design $\mathbf{W}$ that achieves these rates can be obtained via solving the feasibility problem given below,
\begin{subequations} \label{eq:fe_prob}
\begin{align}
\label{eq:fe_prob_a}
\text{find} \hspace{0.5cm} &\left\{ {{\bf{w}}_u^r} \right\}, u \in \mathcal{U}, r \in \mathcal{V}_u \\ 
\label{eq:fe_prob_b}
\text{s.t.} \hspace{0.5cm} &\sum\limits_{u \in \mathcal{U}} {\left\| {{\bf{w}}_u^r} \right\|_2^2}  \le {P_r},\forall r \in \mathcal{R},\\
\label{eq:fe_prob_c}
 &\frac{{{{\left| {{\Psi _u}} \right|}^2}}}{{\sum\limits_{u' \in {\cal U} ,u' \ne u} {{{\left| {{\Psi _{u,u'}}} \right|}^2} + {\sigma ^2}} }}  \ge \gamma _u^*, \forall u \in \mathcal{U},
\end{align}
\end{subequations}where $\gamma _u^* = 2^{R_u^*} - 1$.

The feasibility problem in~\eqref{eq:fe_prob} is \emph{not convex}; however, by exploiting its special structure, we can transform this problem into a SOCP form, which can be solved efficiently. The transformation is presented as follows. Firstly, let ${{\bf{w}}^r}$ be the long column vector such that ${{\bf{w}}^r} = {\left[ {{{\left( {{\bf{w}}_1^r} \right)}^T},{{\left( {{\bf{w}}_2^r} \right)}^T},...{{\left( {{\bf{w}}_U^r} \right)}^T}} \right]^T}$, $\forall r \in \cal R$. Constraint~\eqref{eq:fe_prob_b} can be rewritten in a SOC form as
\begin{equation} \label{eq:soc_1}
{\left\| {{{\bf{w}}^r}} \right\|_2} \le \sqrt {{P_r}} ,\forall r \in {\cal R}.
\end{equation}

%${{\bf{w}}^r} = \left( {{\bf{w}}_1^r,{\bf{w}}_2^r,...{\bf{w}}_U^r} \right)$, $\forall r \in \mathcal{R}$; by using the $vec$ operator, which converts a matrix into a column vector, constraint~\eqref{eq:fe_prob_b} can be rewritten in a SOC form as
%\begin{equation} \label{eq:soc_1}
%\left\| {vec \left( {{{\bf{w}}^r}} \right)} \right\|_2 \le \sqrt {{P_r}}, \forall r \in \mathcal{R}.
%\end{equation}
%
%
%Furthermore, observe that forcing the imaginary part of ${{{\left( {{\bf{h}}_u^r} \right)}^H}{\bf{w}}_u^r}$ to zero does not affect the optimality of (\ref{eq:fe_prob}), i.e., if ${{\bf{w}}_u^r}$ is optimal, then ${\bf{w}}_u^rdiag\left\{ {{e^{j\phi i}}} \right\}$, where $\phi_i, i=1,...N_r$ are arbitrary phases, is also optimal. This is straightforward to verify, since a phase rotation on ${{\bf{w}}_u^r}$ will result in the same objective function and constraints. Hence, we can restrict ourselves to precoders in which ${{{\left( {{\bf{h}}_u^r} \right)}^H}{\bf{w}}_u^r} \ge 0$, $\forall u \in \mathcal{U}, r \in \mathcal{V}_u$, i.e., each product has a nonnegative real part and a zero imaginary part. Using this property, we now rewrite constraint (\ref{eq:fe_prob_c}) into standard form. 
Furthermore, \eqref{eq:fe_prob_c} is equivalent to 
\begin{align} \label{eq:soc_2}
\left( {1 + \frac{1}{{\gamma _u^*}}} \right){\left| {\Psi_u} \right|^2} \ge \sum\limits_{u' \in \mathcal{U}} {{{\left| {\Psi_{u,u'}} \right|}^2} + {\sigma ^2}} ,{\rm{ }}\forall r \in {\cal R}.
\end{align}Since ${\bf{h}}_u^r{\bf{w}}_u^r \ge 0$, as we considered previously, we can take the square root of both sides in \eqref{eq:soc_2}, which yields, 
%\begin{align}  \nonumber
%\Psi_u\sqrt {1 + \frac{1}{{\gamma _u^*}}}  &\ge \sqrt {\sum\limits_{u' \in \cal U} {{{\left| {\Psi_{u,u'}} \right|}^2} + {\sigma ^2}} }= \\\label{eq:soc_3}
%&= {\left\| {\left[ {{\Psi _{u,1}},...{\Psi _{u,U}},\sigma } \right]} \right\|_2}.
%\end{align}
\begin{align}  \label{eq:soc_3}
\Psi_u\sqrt {1 + \frac{1}{{\gamma _u^*}}}  \ge \sqrt {\sum\limits_{u' \in \cal U} {{{\left| {\Psi_{u,u'}} \right|}^2} + {\sigma ^2}} }= {\left\| {\left[ {{\Psi _{u,1}},...{\Psi _{u,U}},\sigma } \right]} \right\|_2}.
\end{align}
%where ${vec\left( {\left\{ {\sum\limits_{r' \in {V_{u'}}} {{{\left( {{\bf{h}}_u^{r'}} \right)}^H}{\bf{w}}_{u'}^{r'}} ,u \in U} \right\},\sigma } \right)}$ is the column vector whose the first $U$ elements are ${\sum\limits_{r' \in {V_{u'}}} {{{\left( {{\bf{h}}_u^{r'}} \right)}^H}{\bf{w}}_{u'}^{r'}} }$, $u = 1,...U$, and the last element is $\sigma$. 
It can be seen that~\eqref{eq:soc_3} follows the SOC form; hence, using~\eqref{eq:soc_1} and \eqref{eq:soc_3}, we are now ready to recast the feasibility problem in~\eqref{eq:fe_prob} in the standard SOCP form as follows,
\begin{subequations} \label{eq:socp_prob}
\begin{align}
\label{eq:socp_prob_a}
&\text{find} \hspace{0.5cm} \left\{ {{\bf{w}}_u^r} \right\}, u \in \mathcal{U}, r \in \mathcal{V}_u \\ 
\label{eq:socp_prob_b}
\text{s.t.} \hspace{0.3cm} &\left\| {{{\bf{w}}^r}} \right\|_2 \le \sqrt {{P_r}}, \forall r \in \mathcal{R},\\
&{\left\| {\left[ {{\Psi _{u,1}},...{\Psi _{u,U}},\sigma } \right]} \right\|_2} \le \Psi_u\sqrt {1 + \frac{1}{{\gamma _u^*}}},
\label{eq:socp_prob_c}
\end{align}
\end{subequations}
The solution ${{\bf{W}}^*}$ for~\eqref{eq:socp_prob} can be obtained using standard SOCP techniques such as the interior-point methods~\cite{ye2011interior} or the SOCP solvers (e.g., CPLEX, MOSEK). In summary, the optimal beamformer design of the CBD problem in~\eqref{eq:static} for a given radio clustering decision $\mathcal{S}$ can be obtained by the procedures described in Algorithm~\ref{algo:CBD}.
\begin{algorithm}\label{CBD}
\caption{Cooperative Beamformer Design (CBD).}\label{algo:CBD}
\begin{algorithmic}
\State \textbf{(1)} Solve the SOCP problem in~\eqref{eq:power2} to find ${{\bf{\tilde W}}}$
%\State \textbf{(2)} Calculate ${{\tilde R}_u}$=${R_u}\left| {_{{\bf{W}}={\bf{\tilde W}}}} \right.$, $\forall u \in \mathcal{U}$
\State \textbf{(2)} Verify constraint~\eqref{eq:static_c}
	\begin{itemize}
		\item \textit{If} $\sum\limits_{u \in {\cal U}} {{R_u}\left( {{\cal S},{\bf{\tilde W}}} \right)}  \le \Omega $, \textit{return} 
		$\bf{W}^*$=$\bf{\tilde W}$. 
		\item \textit{Otherwise}: Drop users' rates using the greedy algorithm
			\begin{itemize}
				\item \textit{Repeat}: Update ${R_{u'}}\left( {S,{\bf{\tilde W}}} \right) = {R_{u'}}\left( {S,{\bf{\tilde W}}} \right) - \tau$, where $\tau$ is small decreasing step and ${q_{u'}}$=$\min \left\{ {{q_u}:{q_u} > 0,u \in \mathcal{U}} \right\}$. Go to the next user when ${R_{u'}}\left( {S,{\bf{\tilde W}}} \right)=0$.
				\item \textit{Until}: $\sum\limits_{u \in {\cal U}} {{R_u}\left( {{\cal S},{\bf{\tilde W}}} \right)}  \le \Omega $
				\item Solve feasibility problem~\eqref{eq:fe_prob}, get ${{\bf{W}}^*}$. \textit{Return}
			\end{itemize}
	\end{itemize}
\end{algorithmic}
\end{algorithm}

\section{Joint Dynamic Radio Clustering and Beamforming Design}\label{sec:design2}
In the previous section, the CBD problem has been transformed into an equivalent SOCP form. As a result, our considered \textit{Dynamic-RC} problem in~\eqref{eq:prob1} can also be transformed into a Mixed-Integer SOCP (MI-SOCP) problem with binary variables $s_u^r$'s. In a network with $U$ users and $R$ RRHs, there are $2^{UR}$ possible clustering patterns. The optimal solution to the clustering decision can be found via exhaustive search or using standard global optimization solvers. However, these approaches usually have a complexity growing exponentially with the problem size, which is not a practical approach. Hence, in this section, we present a method to solve the \textit{Dynamic-RC} problem given in~\eqref{eq:prob1} by \emph{iteratively} solving the CBD problem using Algorithm~\ref{algo:CBD}. In particular, we take advantage of the $l1$-norm reweighting technique to adjust the approximation of the clustering variables after each iteration. 

Firstly, given the relationship of $s_u^r$ and ${\bf{w}}_u^r$, we can represent $s_u^r$ by $l0$-norm expression of ${\bf{w}}_u^r$ as follows,
\begin{equation} \label{eq:l0}
s_u^r = {\left\| {\left\| {{\bf{w}}_u^r} \right\|_2^2} \right\|_0}.
\end{equation}The above expression allows us to leverage the $l1$-norm reweighting technique, which has been effectively applied in the literature to approximate the $l0$-norm~\cite{l1_reweighted}, i.e., ${\left\| \chi  \right\|_0} \approx \sum\limits_k {{\rho _k}{\chi _k}}$,
%\begin{equation} \label{eq:chi}
%{\left\| \chi  \right\|_0} \approx \sum\limits_k {{\rho _k}{\chi _k}},
%\end{equation}
where $\chi  \in {\mathbb{R}^n}$ and $\rho _1,\rho _2,...,\rho _n$ are positive weights. With $n=1$, by choosing $\chi _k=\left\| {\bf{w}}_u^r \right\|_2^2$, we get $s_u^r \approx \rho _u^r\left\| {{\bf{w}}_u^r} \right\|_2^2$,
%Substituting~\eqref{eq:chi} into \eqref{eq:l0}, where $\chi _k=\left\| {\bf{w}}_u^r \right\|_2^2$, we get 
%\begin{equation} \label{eq:su_rho}
%s_u^r \approx \rho _u^r\left\| {{\bf{w}}_u^r} \right\|_2^2,
%\end{equation}
in which the weight $\rho_u^r$ is adjusted iteratively as
\begin{equation} \label{eq:rho_update}
\rho _u^r = \frac{1}{{\left\| {{\bf{\hat w}}_u^r} \right\|_2^2 + \epsilon }}, \forall u \in \mathcal{U}, r \in \mathcal{R},
\end{equation}with ${\left\| {\bf{\hat w}_u^r} \right\|_2^2}$ obtained from the previous iteration. In~\eqref{eq:rho_update}, the parameter $\epsilon$ is a very small positive number introduced to provide stability and to ensure that in case $\left\| {\bf{w}}_u^r \right\|_2^2 = 0$, it does not strictly prohibit a non-zero estimate in the next iteration. 

%\begin{figure*}[t]
% \centering
%\includegraphics[scale = 1]{fig/scenarios.eps} 
%\caption{Weighted Sum-Rate System Utility (WSRSU) of a CRAN downlink system using different radio cooperation schemes.}\label{fig:optimal_dynamic}
%\end{figure*}

The \textit{Dynamic-RC} problem in~\eqref{eq:prob1} -- given now $\rho_u^r$'s -- can be rewritten as,
\begin{subequations} \label{eq:approx_prob}
\begin{align}
\label{eq:approx_prob_a}
\mathop {\max }\limits_{{\bf{w}}_u^r,r \in \mathcal{R},u \in \mathcal{U}}  &\sum\limits_{u \in \mathcal{U}} {{q_u}{R_u}}\left( {{\cal S},{\bf{W}}} \right) \\
\label{eq:approx_prob_b}
\text{s.t.} \hspace{0.5cm} &\sum\limits_{u \in \mathcal{U}} {\left\| {{\bf{w}}_u^r} \right\|_2^2}  \le {P_r},\forall r \in \mathcal{R},\\
\label{eq:approx_prob_d}
&\sum\limits_{u \in \mathcal{U}} {{R_u}\left( {{\cal S},{\bf{W}}} \right)}  \le \Omega ,\\
\label{eq:approx_prob_e}
&\sum\limits_{u \in \mathcal{U}} {\rho _u^r\left\| {{\bf{w}}_u^r} \right\|_2^2}  \le {N_r}.
\end{align}
\end{subequations}Note that constraint~\eqref{eq:approx_prob_e} can be written in SOC form as,
%\begin{equation} \label{eq:soc_constraint}
%{\left\| {vec\left( {\left\{ {{\bf{w}}_u^r\sqrt {\rho _u^r} ,u \in \mathcal{U},r \in \mathcal{R}} \right\}} \right)} \right\|_2} \le \sqrt {{N_r}} ,\forall r \in {\cal R}.
%\end{equation}
\begin{equation} \label{eq:soc_constraint}
{\left\| {\left[ {{\bf{w}}_1^r\sqrt {\rho _1^r} ,...,{\bf{w}}_U^r\sqrt {\rho _U^r} } \right]} \right\|_2} \le \sqrt {{N_r}}, \forall r \in {\cal R}.
\end{equation}
Thus, the problem in~\eqref{eq:approx_prob} is similar to the CBD problem in~\eqref{eq:static} with the additional SOC constraint~\eqref{eq:soc_constraint}, which can be solved efficiently using Algorithm~\ref{algo:CBD}. To clarify the idea, we present the iterative method to solve the \textit{Dynamic-RC} problem in Algorithm~\ref{algo:iterative} below.

\begin{algorithm}
\caption{Dynamic Radio Cooperation via Iterative SOCP}\label{algo:iterative}
\begin{algorithmic}
\State \textbf{(1) Initialization:} set $\rho _u^r=0$, $\forall u \in \mathcal{U}, r \in \mathcal{R}$
\State \textbf{(2) Iteration:} 
	\begin{itemize}
		\item[a)] Solve problem~\eqref{eq:approx_prob} with the current value of $\rho _u^r$ using Algorithm~\ref{algo:CBD}. In particular, Step~(1) in Algorithm~\ref{algo:CBD} will solve problem~\eqref{eq:power2} with the additional constraint~\eqref{eq:soc_constraint}.
		\item[b)] Update the weights $\rho_u^r$'s using the solution ${\bf{\hat w}}_u^r$'s obtained in the previous step as,
		\begin{align} \label{eq:update_weight}
			\rho _u^r={\left( {\left\| {{\bf{\hat w}}_u^r} \right\|_2^2 + \epsilon } \right)^{ - 1}}, \forall u \in \mathcal{U}, r\in \mathcal{R}.
		\end{align}		
	\end{itemize}
\State \textbf{(3) Check convergence:} Repeat Step~(2) until convergence or the max number of iterations is reached. 
\end{algorithmic}
\end{algorithm}
Note that RRH $r$ is included in the serving cluster of user $u$, i.e., $r \in \mathcal{V}_u$, if the beamformer from RRH $r$ to user $u$, ${\bf{w}}_u^r$, is nonzero. Since $\rho_u^r = 0$ in the first iteration in Algorithm 2, the constraint (\ref{eq:soc_constraint}) is automatically satisfied. Thus, initially each RRH can be selected into more than $N_r$ clusters. After that, the weights $\left\{ {\rho _u^r} \right\}$ are updated inversely proportional to the beamforming power as in (\ref{eq:update_weight}). Therefore, among the beamformers from all the RRHs to a target user, those with highest powers are most likely to be identified as nonzero in the next iteration.
This allows for successive better estimation of the clustering decision, i.e., identifying the nonzero beamformers from RRHs to users. As will be shown later in our simulation results, the beamforming powers quickly converge within a few iterations.

\emph{\underline{Complexity analysis}: }The computational complexity of Algorithm~\ref{algo:iterative} mainly lies in Step~(2a) where a SOCP problem is solved. Assuming the same number of antennae $N_r$ on the RRHs, the total number of variables in this SOCP problem is $URN_r$, where $U$ and $R$ are the numbers of users and RRHs. Thus, the computational complexity of the interior-point method to solve such a SOCP problem is approximately  ${\cal O}\left( {{{\left( {UR{{N}_r}} \right)}^{3.5}}} \right)$~\cite{ye2011interior}. This is significantly advantageous for a large network compared to the optimal design using existing solvers, which are characterized by a prohibitively exponential-time complexity.

Furthermore, in practical networks, a RRH $r$ should not be included in the serving cluster of user $u$ if $r$ is very far away from $u$. Assuming a network of hexagonal cells, we can pre-select only the 7 RRHs having strongest channel coefficients to user $u$ to be the candidate serving cluster of user $u$, denoted as $\mathcal{C}_u$. After the pre-selection process, Algorithm 2 will identify the optimal serving cluster $\mathcal{V}_u$ within the subsets of $\mathcal{C}_u$. This can significantly reduce the complexity of Algorithm 2 to ${\cal O}\left( {{{\left( {7U{{N_r}}} \right)}^{3.5}}} \right)$. We adopt the pre-selection of serving cluster candidates in the simulation and numerical results show that this approach performs very close to the optimal solution.

%Furthermore, observe that in a practical network, a RRH $r$ should not be included in the serving cluster of user $u$ if $r$ is very far away from $u$. Assuming a network of hexagonal cells, for each user $u$ in the cell covered by RRH $r$, it is reasonable to only consider $r$ and its 6 neighboring RRHs to be the candidates of the serving cluster of user $u$. By choosing $\mathcal{V}_u$ to be the subset of 7 RRHs closest to user $u$, we can significant reduce the complexity of Algorithm 2 to approximately ${\cal O}\left( {{{\left( {7U{{\bar N}}} \right)}^{3.5}}} \right)$, where ${{{\bar N}}}$ is the average number of antennae on each RRH. In fact, we adopt this in our simulation and the numerical results show that this approach performs very close to the optimal solution.

%+++++++++++++++++++++++++++++++++++++++++++++++++++++++++++++++++
\section{Performance Evaluation}\label{sec:result}
%\begin{figure}[t]
%\centering
%\includegraphics[scale = 0.5]{fig/7cells_new.eps} 
%\caption{Example of cellular network of 7 hexagonal cells with a non-uniform user distribution.}\label{fig:7cell_non_uniform}
%\end{figure}

In this section, simulation results are presented to evaluate the performance of our proposed \textit{Dynamic-RC} algorithm. We consider a network of hexagonal cells with a RRH in the center of each cell. The neighboring RRHs are separated $1~\mathrm{Km}$ apart from each other. We assume that all the wireless channels in the system experience \emph{block fading} such that the channel coefficients stay constant during each scheduling interval but can vary from interval to interval, i.e., the \emph{channel coherence time} is not shorter than the scheduling interval. We assume that all the RRHs have the same number of transmit antennae $N_r$ and transmit power budget $P_r$. The channel coefficients are calculated following the path-loss model, given as $L~\mathrm{[dB]}=148.1+37.6\log_{10}d_{\mathrm{[km]}}$, and the log-normal shadowing variance set to $8~\mathrm{dB}$. In addition, it is assumed that the channel bandwidth $B$ is $10~\mathrm{MHz}$, is reused across all the users, and the noise spectral density is $-100~\mathrm{dBm/Hz}$. 
%{\color{red}
%Suppose that there are $U$=$16$ users randomly distributed in the network in a non-uniform manner, as depicted in the example in Fig.~\ref{fig:7cell_non_uniform}. The utility marginal functions $q_u$'s are chosen randomly such that $0<q_u \leq 1, \forall u \in \mathcal{U}$. }
%\begin{figure}[t] 
% \centering
%\includegraphics[scale = 0.6]{fig/3cases.eps} 
%\caption{Different user distribution scenarios: Scenario 1 (uniform) with all \emph{medium} (loaded) cells; Scenario 2 (uneven), \emph{light} and \emph{heavy} (loaded) cells are intermixed together; Scenario 3 (extremely uneven), \emph{heavy} cells are grouped together, and the heavy cell group is surrounded by \emph{light} cells. }\label{fig:scenarios}
%\end{figure}

\begin{figure*}[t]
 \centering
 \begin{tabular}{ccc}
\hspace*{-.6cm}\includegraphics[scale = .6]{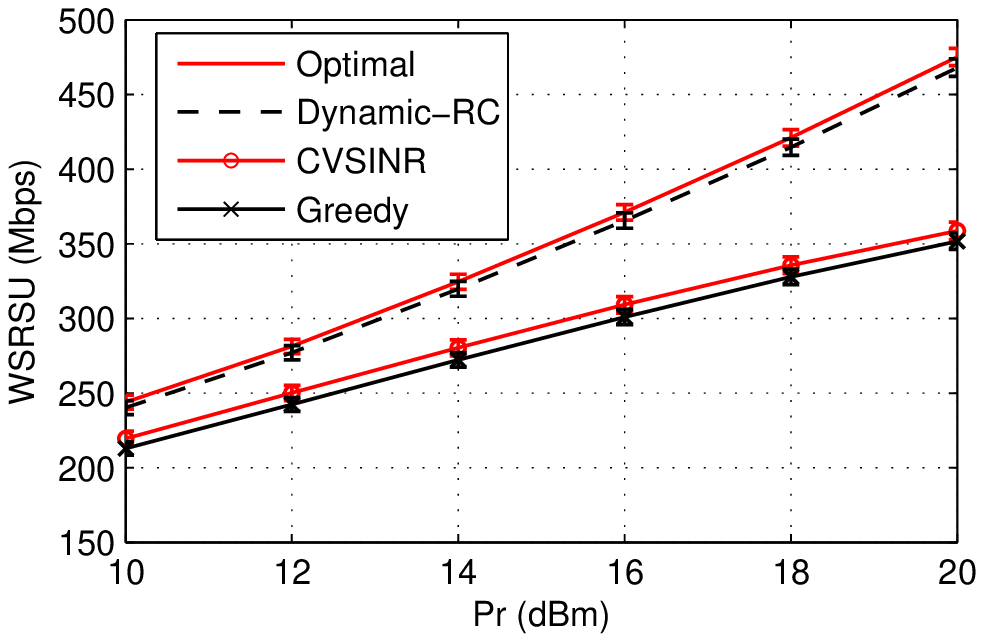} &
\hspace*{-.6cm}\includegraphics[scale = .6]{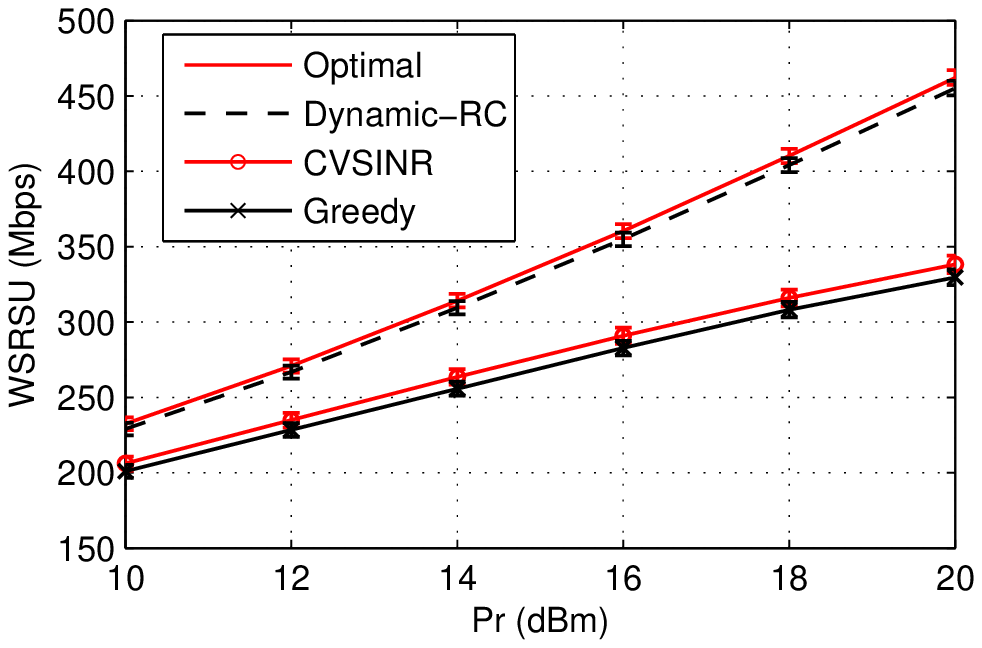} &
\hspace*{-.6cm}\includegraphics[scale = .6]{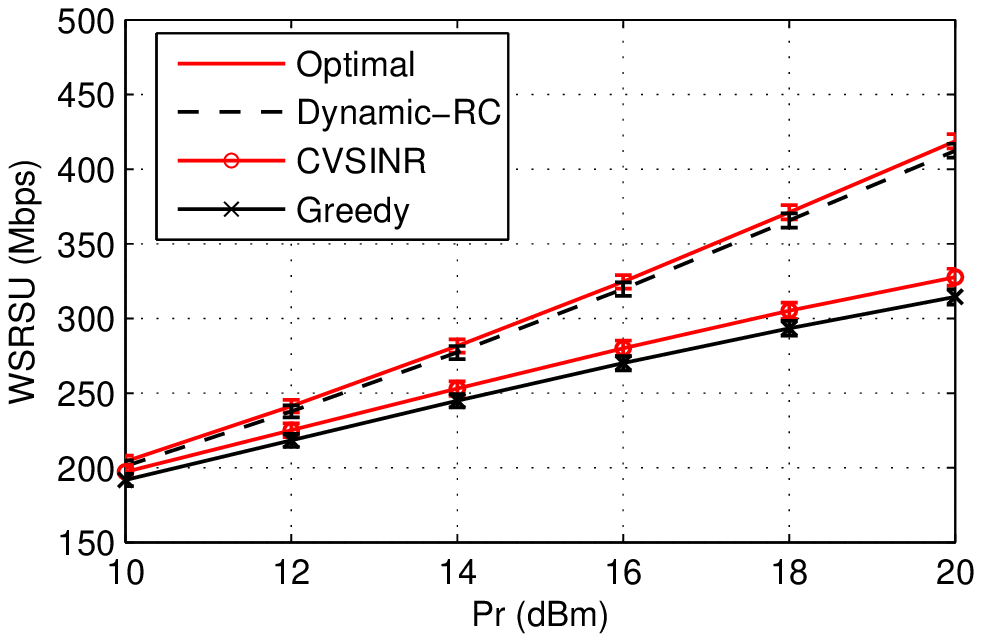} \\
\small (a) & \small(b) & \small(c)
\end{tabular}
\caption{ Weighted Sum-Rate System Utility (WSRSU) of a C-RAN downlink system using different radio cooperation schemes, evaluating on three different user distribution scenarios. (a)-Scenario 1, (b)-Scenario 2, (c)-Scenario 3. 
}\label{fig:opt_dyn_cvsinr}
\end{figure*}

{\underline{\textit {WSRSU performance:  }}}
Firstly, we consider a system without the computing-resource constraint and evaluate the performance of the four radio cooperation algorithms below.
\begin{itemize}
\item \textit{Optimal}: The WSRSU of the optimal scheme is obtained by using the solver MOSEK to solve the equivalent MI-SOCP presentation of problem~\eqref{eq:prob1}.
\item \textit{Dynamic-RC}: Our proposed dynamic radio cooperation, where the solutions are obtained from our iterative, low-complexity Algorithm~\ref{algo:iterative}. 

\item \textit{CVSINR}: A downlink cooperation scheme proposed in \cite{gong2011joint} where the cluster for each user is formed heuristically based on the relative signal strength and the clustered virtual SINR (CVSINR) algorithm is used to design the beamforming vectors.

\item \textit{Greedy}: A greedy clustering algorithm proposed in \cite{baracca2012dynamic}, which solves an equivalent set covering problem to select the set of non-overlapping base station clusters. This scheme uses zero-forcing as the criterion to design beamformers and a greedy algorithm is used for user scheduling.  
\end{itemize} 

\begin{figure}[t] 
 \centering
\includegraphics[scale = 0.6]{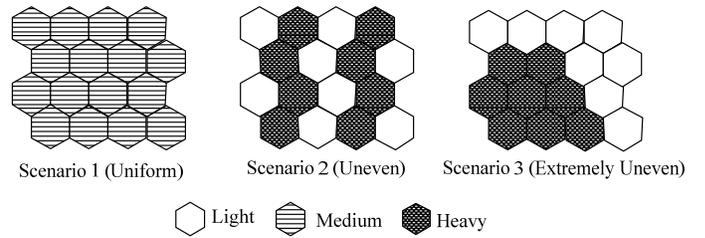} 
\caption{Different user distribution scenarios: Scenario 1 (uniform) with all \emph{medium} (loaded) cells; Scenario 2 (uneven), \emph{light} and \emph{heavy} (loaded) cells are intermixed together; Scenario 3 (extremely uneven), \emph{heavy} cells are grouped together, and the heavy cell group is surrounded by \emph{light} cells. }\label{fig:scenarios}
\end{figure}

We evaluate the four schemes above in a network of 16 cells with three different user distribution scenarios as shown in Fig.~\ref{fig:scenarios}. In particular, \emph{Scenario 1} consists of all \emph{medium} (loaded) cells where users are distributed uniformly over all the cells; \emph{Scenarios 2} and \emph{Scenarios 3} consist of \emph{light} and \emph{heavy} (loaded) cells, however the \emph{heavy} cells are
intermixed with \emph{light} cells in \emph{Scenarios 2} to represent micro-tidal effect while they are grouped together in \emph{Scenarios 3} to represent the macro-tidal effect. In our simulation, we perform 500 drops, in each drop 32 users are placed randomly in the network with 1 user in a light cell, 2 users in a medium cell and 3 users in a heavy cell. The utility marginal functions $q_u$'s are chosen randomly such that $0<q_u \leq 1, \forall u \in \mathcal{U}$. 

Fig.~\ref{fig:opt_dyn_cvsinr}-(a), (b), (c) plot the WSRSU performance of the four considered radio cooperation schemes in Scenario- 1, 2, 3, respectively. It can be seen that our proposed \textit{Dynamic-RC} scheme and the \textit{Optimal} scheme significantly outperform the \emph{CVSINR} and \emph{Greedy} schemes in all three scenarios. This is because the heuristic clustering of the RRHs in the later two schemes is suboptimal, plus their beamforming design algorithms only aim to minimize the intra-cluster interference \emph{but not} the inter-cluster interference. On the other hand, our proposed \emph{Dynamic-RC} scheme takes into account the global network condition that is available at the VBS pool, which provides better clustering decision and beamforming design. Compared to the \textit{optimal} scheme, our proposed \textit{Dynamic-RC} strategy via Algorithm~\ref{algo:iterative} shows a small loss in WSRSU performance but has a significant advantage in reducing the execution time. In fact, in our simulation for the considered system configuration ($U$=$32$, $R$=$16$), MOSEK solver takes more than $100~\mathrm{s}$ to obtain the optimal solution of the MI-SOCP problem, while each iteration in Algorithm~\ref{algo:iterative} takes less than a second and the algorithm overall converges within $15$ iterations.

{\underline{\textit {Impact of Maximum Cluster Size:  }}}
Fig.~\ref{fig:cdf}-(a), (b) plot the CDF of average user rate (w.r.t. 32 users) achieved by Dynamic-RC scheme with different choices of the maximum cluster size, $V_{max}$. In each case, only $V_{max}$ RRHs having the strongest channel coefficients to a user are chosen to be the candidates of that user's serving cluster. This pre-selection is done before running Algorithm 2 to finally find the best serving cluster for each user. When $V_{max} = 1$, there is no cooperation among the RRHs. The results in Fig.~\ref{fig:cdf}-(a), (b) are obtained by performing 500 drops on \emph{Scenario 2} and \emph{Scenario 3}, respectively, with Pr =$\mathrm{10 dBm}$. The utility marginal functions are updated in each drop according to the proportional fairness criterion, i.e., ${q_u} = 1/{\bar R_u}$ where ${\bar R_u}$ is the long term average data rate for user $u \in \mathcal{U}$. It can be seen that the improvement in average user rate due to larger cluster size in Scenario 3 (macro-tidal effect) is greater than that of Scenario 2 (micro-tidal effect). For example, the dynamic cooperation scheme with $V_{max} = 3,5,7$ provides 130$\%$, 137$\%$, and 138.6$\%$ gain, respectively, for the 60th-percentile average user rate over the non-cooperation scheme ($V_{max} = 1$), in Scenario 2; while the corresponding gains in Scenario 3 are 145$\%$, 159$\%$, and 162$\%$, respectively. Although not included here due to space limitation, we observe that when $V_{max}$ exceeds 7 cells, the additional gain is negligible.
 
%larger clusters offer better average user rates, and the improvement 
%
%with the price of higher complexity. {\color{red}
%The improvement in average user rates is more obvious in Scenario 3 (Fig.~\ref{fig:cdf}-b). Although not included here due to space limitation, we observe that when $V_{max}$ exceeds 7 cells, the gain is negligible.}   

\begin{figure}[t]
 \centering
 \begin{tabular}{c}
\includegraphics[scale = 0.62]{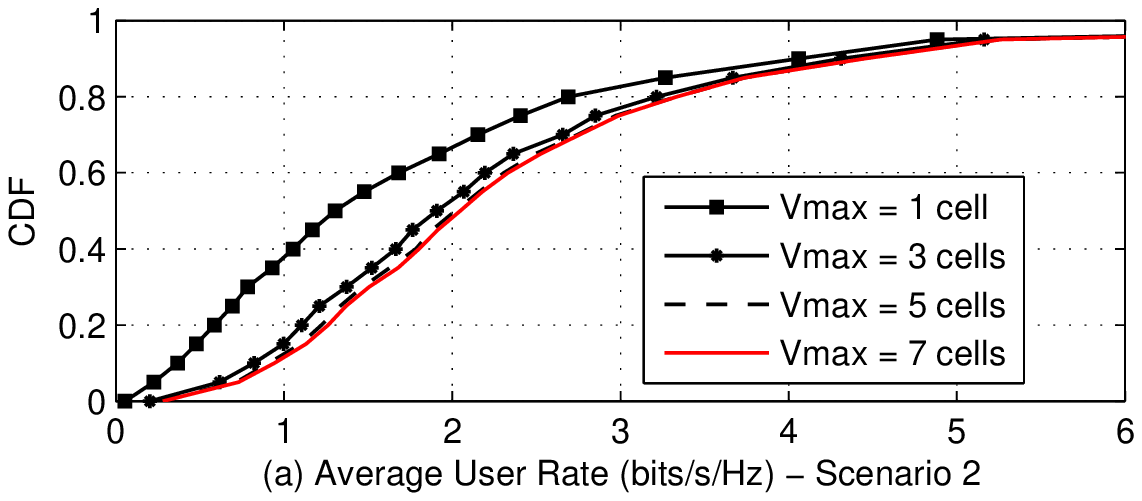} \\
\includegraphics[scale = 0.62]{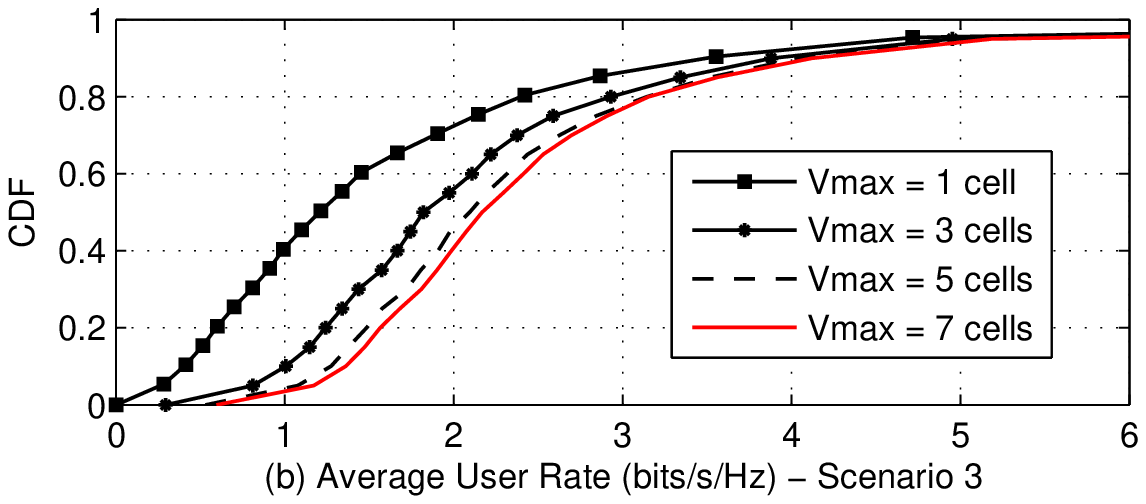} \\
\end{tabular}
\caption{CDF of Average User Rate obtained by Dynamic-RC scheme with different numbers of the maximum cluster size. (Pr = 10 dBm)}\label{fig:cdf}
\end{figure}

%\begin{table}[htb!] {\caption{50th-percentile Average User Rate [bits/s/Hz].}
%\label{tab:para2}}
%\begin{center}
%\begin{tabular}{|c|c|}
%  \hline
%  % after \\: \hline or \cline{col1-col2} \cline{col3-col4} ...
%  Scheme & Description \\
%  \hline
%  \textit{OPA-S} & Capacity\\
%  \hline
%  \shortstack{\textit{EPA-S}\\\;} & \shortstack{$S$}\\
%    \hline
%\end{tabular}
%\end{center}
%\end{table}

{\underline{\textit{ Benefits of Computing Resource Sharing:  }}}
To evaluate the impact of the computing-resource constraint on the system performance, Fig.~\ref{fig:compuing_constraint} compares the WSRSU performance of our considered system with the  \textit{centralized} computing-resource constraint, as expressed in~\eqref{eq:centralized}, versus a conventional system with a \textit{distributed} computing-resource constraint, as expressed in~\eqref{eq:distributed}. In particular, we consider a network of 4 cells with 2 users in each cells in random locations and $q_u$'s are chosen randomly. For a fair comparison, we set $\arg \Gamma(C)$ to $400~\mathrm{Mbps}$ and $\arg \Gamma(C_r)$ to $100~\mathrm{Mbps}$, and ran the \textit{Dynamic-RC} scheme in Algorithm~\ref{algo:iterative} on both systems. Note that, in this setting, each of the 4 RRHs in the \textit{distributed} system is provisioned to process maximum $100~\mathrm{Mbps}$ of user baseband traffic at a time, while in the \textit{centralized} system the VBS pool is provisioned to process maximum $400~\mathrm{Mbps}$ baseband traffic at a time. We say that the computing resource is saturated in each system when the achieved sum-rate (SR) of all the users reaches the maximum provisioned processing traffic rate. As the transmit power increases, observe in Fig.~\ref{fig:compuing_constraint} that the computing capacity of the VBS pool in the \emph{centralized} system saturates earlier than the total computing capacity of the distributed system does (when the computing capacity is saturated at all the RRHs). In fact, the WSRSU and SR of the distributed system saturate almost at the same time while the WSRSU of the centralized system continues to increase after the saturation point (of the SR), and is significantly higher (up to $250\%$ gain) than that of the distributed system. This demonstrates the great potential gains of C-RANs using our \emph{Dynamic-RC} scheme over the conventional distributed RANs in terms of WSRSU, computing resource and transmit power utilization.

%\begin{figure}[t]
% \centering
%\includegraphics[scale = 0.55]{fig/cdf_1357.eps} 
%\caption{CDF of Average User Rate obtained by Dynamic-RC scheme with different numbers of the maximum cluster size.}\label{fig:cdf}
%\end{figure}

%\begin{figure*}[t]
% \centering
% \begin{tabular}{ccc}
%\hspace*{-.6cm}\includegraphics[scale = .6]{fig/16_8_2_uni.eps} &
%\hspace*{-.6cm}\includegraphics[scale = .6]{fig/16_8_2_dist1.eps} &
%\hspace*{-.6cm}\includegraphics[scale = .6]{fig/16_8_2_dist2.eps} \\
%\small (a) & \small(b) & \small(c)
%\end{tabular}
%\caption{ Weighted Sum-Rate System Utility (WSRSU) of a C-RAN downlink system using different radio cooperation schemes, evaluating on three different user distribution scenarios. (a)-Scenario 1, (b)-Scenario 2, (c)-Scenario 3. 
%}\label{fig:opt_dyn_cvsinr}
%\end{figure*}

\begin{figure}[t]
 \centering
\includegraphics[scale = 0.62]{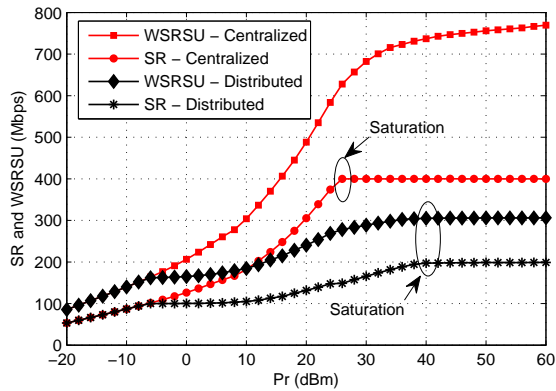} 
\caption{Downlink WSRSU of a C-RAN system with centralized computing resource (\emph{Centralized}) versus that of a traditional system with distributed computing resources (\emph{Distributed}).}\label{fig:compuing_constraint}
\end{figure}

{\underline{\textit {Convergence Behavior of Algorithm 2:  }}}
Fig.~\ref{fig:power_scale} illustrates the convergence behavior of Algorithm~\ref{algo:iterative} in identifying the RRH cluster for a user. We choose randomly a user $u^*$ and monitor the beamforming powers from the 7 candidate RRHs for it's serving cluster. The evolution of the beamforming powers in $\mathrm{dBm/Hz}$ from these RRHs to user $u^*$, calculated as $\left\| {{\bf{w}}_{u^*}^r} \right\|_2^2$, $r=1,...,7$, is shown in Fig.~\ref{fig:power_scale}. Observe that after the $4$-th iteration, only the beamformers from RRH~1 and RRH~4 maintain a non-trivial power, while the rest are forced to almost zero. In this case, the optimal serving cluster of user $u^*$ is identified to be $\mathcal{V}_{u^*}=\left\{\text{RRH 1},\text{RRH 4}\right\}$ within only a few iterations, which demonstrates the efficiency of our proposed \emph{Dynamic-RC} algorithm in quickly making the clustering decision and beamforming design.

\begin{figure}[t]
 \centering
\includegraphics[scale = 0.62]{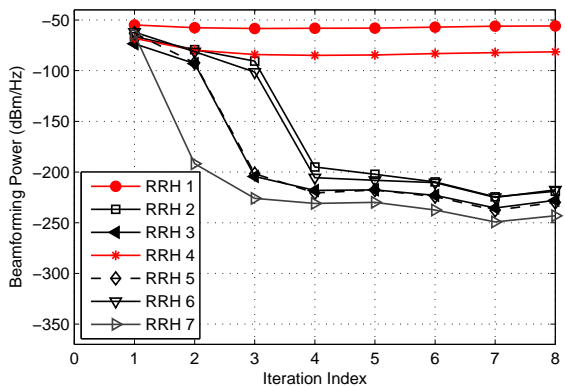} 
\caption{Convergence behavior of the beamforming power at 7 candidate RRHs for the serving cluster of an example user.}\label{fig:power_scale}
\end{figure}
%+++++++++++++++++++++++++++++++++++++++++++++++++++++++++++++++++++
\section{Conclusions and Future Work} \label{sec:con}
We proposed a novel dynamic radio cooperation strategy for Cloud Radio Access Networks (C-RANs) that takes advantage of real-time communication and computing-resource sharing among Virtual Base Stations (VBSs). The underlying optimization problem was formulated as a mixed-integer non-linear program, which is NP-hard. Our approach transforms the original problem into a Mixed-Integer Second-Order Cone Program (MI-SOCP) that is efficiently solved using a novel low-complexity, iterative algorithm. Simulation results showed that our low-complexity algorithm provides close-to-optimal performance in terms of weighted sum-rate system utility while significantly outperforming conventional radio clustering and beamforming schemes.

\textbf{Future Work: } The goal of our future work is to address the system-related issues and evaluate the feasibility and performance of the proposed strategy in a practical system. In fact, we are implementing a C-RAN testbed which consists of an open-source LTE platform OpenAirInterface running on a general-purpose desktop server to realize the VBS pool, and a number of USRP B210/X310 boards to realize the RRHs. 

\textbf{Acknowledgment: }This work was supported in part by the National Science Foundation Grant No. CNS-1319945
%+++++++++++++++++++++++++++++++++++++++++++++++++++++++++++++++++++++++++

%\appendices

%\break
%\newpage

%+++++++++++++++++++++++++++++++++++++++++++++++++++++++++++++++++++++++++
\balance
%\bibliographystyle{ieeetr}\small
%\bibliography{secon14}

\begin{small}
\bibliographystyle{ieeetr}
%\vspace{1cm}
\bibliography{MASS_Harry_CamReady}  
\end{small}

\end{document}